\documentclass{article}
\usepackage[algoruled,linesnumbered,algo2e,vlined]{algorithm2e}
\usepackage{fullpage}
\usepackage{amsmath,amssymb} 
\usepackage{graphicx}
\usepackage{wrapfig}
\newcommand{\ignore}[1]{}
\newcommand{\Occ}{\mathit{Occ}}

\newcommand{\VarOcc}{\mathit{vOcc}}
\newcommand{\suffix}{\mathit{suf}}
\newcommand{\prefix}{\mathit{pre}}
\newcommand{\prefDP}{\mathrm{P}}
\newcommand{\sufDP}{\mathrm{S}}

\newcommand{\derive}{\mathit{val}}

\newcommand{\height}{\mathit{height}}

\newcommand{\prefTrunc}[2]{{^{[#2]}{#1}}}
\newcommand{\sufTrunc}[2]{{#1}^{[#2]}}

\newcommand{\trPrePath}{\mathit{trPrePath}}
\newcommand{\trSufPath}{\mathit{trSufPath}}

\newcommand{\trPre}{\mathit{trPre}}
\newcommand{\trSuf}{\mathit{trSuf}}

\newcommand{\witv}{\mathit{witv}}
\newcommand{\trVarOcc}{\mathit{trvOcc}}
\newcommand{\trPreVarOcc}{\mathit{trPrevOcc}}
\newcommand{\trSufVarOcc}{\mathit{trSufvOcc}}
\newcommand{\dVarOcc}{\mathit{dvOcc}}
\newcommand{\aVarOcc}{\mathit{avOcc}}
\newcommand{\trPretr}{\mathit{trPreAnc}}
\newcommand{\trSuftr}{\mathit{trSufAnc}}

\newcommand{\trR}{\mathit{trR}}
\newcommand{\trSsum}{\mathit{trSsum}}

\newcommand{\leaf}{\mathit{leaf}}

\newtheorem{theorem}{Theorem}
\newtheorem{lemma}[theorem]{Lemma}
\newtheorem{problem}{Problem}

\newcommand{\qed}{\hfill\rule{1ex}{1em}\penalty-1000{\par\medskip}}
\date{}

\begin{document}
\title{Computing $q$-gram Frequencies on Collage~Systems}
\author{
  Keisuke Goto, Hideo Bannai, Shunsuke Inenaga, Masayuki~Takeda
}
\maketitle
\begin{center}
Department of Informatics, Kyushu University, Japan\\
\texttt{\{keisuke.gotou,bannai,inenaga,takeda\}@inf.kyushu-u.ac.jp}
\end{center}

\begin{abstract}
  Collage systems are a general framework for representing 
  outputs of various text compression algorithms.
  We consider the all $q$-gram frequency problem on
  a compressed string represented as a collage system,
  and present an $O((q+h\log n)n)$-time $O(qn)$-space algorithm for calculating
  the frequencies for all $q$-grams that occur in the string.
  Here, $n$ and $h$ are respectively the size and height of the
  collage system.
\end{abstract}
\section{Introduction}
Due to the ever increasing size of data that we generate and utilize,
data is often stored in compressed form.
Since merely decompressing such large scale data can be demanding,
methods for processing compressed strings {\em as is}, that is,
processing a given compressed string without explicitly decompressing it,
has been gaining attention~\cite{lifshits07:_proces_compr_texts,navarro07:_compr,hermelin09:_unified_algor_accel_edit_distan,matsubara_tcs2009,inenaga09:_findin_charac_subst_compr_texts,goto10:_fast_minin_slp_compr_strin,philip11:_random_acces_gramm_compr_strin}.
An interesting property of these methods is that they can be
theoretically 
-- and sometimes even practically -- faster than algorithms which work
on an uncompressed representation of the same data.

{\em Collage systems}~\cite{KidaCollageTCS}
are a general framework to describe compressed representation of
strings, using grammar-like variable assignments.
The basic operations are concatenation, repetition, and truncation.
Collage systems can model outputs of various
compression algorithms~\cite{KidaCollageTCS} such as
grammar based compression algorithms
(e.g.~\cite{SEQUITUR,LarssonDCC99})
and those of the LZ-family (e.g.~\cite{LZ77,LZ78}).
By considering collage systems, it is possible to develop general processing
algorithms which can work on compressed strings generated by any of
these compression algorithms.

In this paper, we consider the problem of determining the frequencies
of all $q$-grams occurring in a string $T$, given a 
collage system representing $T$.
The problem was previously considered for {\em
  regular} collage systems (or equivalently, straight line
programs~(SLPs)~\cite{NJC97}), 
which are collage systems that contain neither truncation nor repetition:
In~\cite{inenaga09:_findin_charac_subst_compr_texts},
an $O(|\Sigma|^2n^2)$ time and $O(n^2)$ space algorithm was
presented for $q=2$,
where $|\Sigma|$ denotes the alphabet size and
$n$ is the size of the SLP.
More recently, a much simpler and more efficient $O(qn)$
time and space algorithm for general $q \geq 2$ was 
developed and was shown to be practically faster than an
algorithm working on {\em uncompressed} strings, 
when $q$ is small~\cite{goto10:_fast_minin_slp_compr_strin}.

The main contribution of this paper is
an $O((q+h\log n)n)$-time and $O(qn)$-space algorithm that
computes the frequencies for all $q$-grams that occur in a given string represented as a
collage system, where $n$ is the size of the collage system,
and $h \leq n$ is the height of the derivation tree of the collage system.
The algorithm is a non-trivial extension of the algorithm of~\cite{goto10:_fast_minin_slp_compr_strin}
so that it can deal with repetitions and truncations.
Given a collage system of size $n$ which describes a string $T$,
it is possible to construct an SLP of size $O(nh\log n)$ which describes the same string $T$.
We can then apply the algorithm of~\cite{goto10:_fast_minin_slp_compr_strin}
to the SLP, achieving an $O(qnh\log n)$-time $O(qnh\log n)$-space solution.
The new $O((q+h\log n)n)$-time $O(qn)$-space solution improves on that.

General collage systems allow for more powerful compression schemes,
for example, while an LZ77 encoded representation of size $m$ 
{\em with self-referencing} may require
$O(m^2\log m)$ size when represented as an SLP, 
it can be represented as a collage system of size $O(m\log m)$~\cite{GasieniecSWAT96}.

\section{Preliminaries}
\subsection{Strings}
Let $\Sigma$ be a nonempty finite set of symbols called the {\em alphabet}.
An element of $\Sigma^*$ is called a {\em string}.
The length of a string $T$ is denoted by $|T|$. 
The empty string $\varepsilon$ is a string of length 0,
namely, $|\varepsilon| = 0$.
For a string $T = XYZ$, $X$, $Y$ and $Z$ are called
a \emph{prefix}, \emph{substring}, and \emph{suffix} of $T$, respectively.
The $i$-th character of a string $T$ is denoted by $T[i]$
for $1 \leq i \leq |T|$, and the substring of a string $T$ that
begins at position $i$ and ends at position $j$ is denoted by
$T[i:j]$ for $0 \leq i \leq j \leq |T|-1$.
For convenience, let $T[i:j] = \varepsilon$ if $j < i$.
For any string $X$, let $X^0 = \varepsilon$ and for any integer $p \geq 1$,
let $X^p = X^{p-1}X$.
For strings $T$ and $P$, let $\Occ(T,P) = \{ i \mid T[i:i+|P|-1] = P\}$ denote the set of occurrences of $P$ in $T$.
For string $T$ and integer $k \geq 1$, 
let $\prefix(T,k)  = T[1:\min\{k, |T|\}]$
and $\suffix(T, k) = T[|T|- \min\{k, |T|\} + 1:|T|]$,
i.e., respectively the prefix and the suffix of $T$ of length at 
most $k$.

\subsection{Collage Systems}
We consider strings described by \emph{collage system}s,
proposed in~\cite{KidaCollageTCS}.
Collage systems are a general framework for representing outputs
of various compression algorithms.
A collage system $\mathcal{T}$ is a set of assignments
$\{  X_1 = expr_1, X_2 = expr_2, \ldots, X_n = expr_n \}$,
where each $X_i$ is a variable and each $expr_i$ is an expression:
\begin{equation*}
  expr_i =
  \begin{cases}
    a           & (a\in\Sigma), \hfill \mbox{(terminal symbol)}\\
    X_{\ell} X_r & (\ell,r < i ), \hfill \mbox{(concatenation)}\\
    (X_s)^p     & (s < i, p > 2), \hfill \mbox{(repetition)}\\
    \prefTrunc{X_s}{k} & (s < i, 1 \leq k < |\derive(X_s)|),\quad \hfill \mbox{(prefix truncation)}\\
    \sufTrunc{X_s}{k}  & (s < i, 1 \leq k < |\derive(X_s)|),  \hfill \mbox{(suffix truncation)}\\
  \end{cases}
\end{equation*}
where $\derive$ is a function defined below.
To simplify the presentation, our definition of collage systems
differs from the original in that we only consider a single variable
$X_n$ for the sequence part.

A collage system is said to be {\em truncation-free} if no prefix truncation nor
suffix truncation is used.
A collage system is said to be {\em regular}, 
if it is truncation-free, and no repetition is used.
(Regular collage systems are equivalent to 
straight line programs (SLPs)~\cite{NJC97}, a general framework
for grammar-based compression.)
Output of the SEQUITUR~\cite{SEQUITUR} and REPAIR~\cite{LarssonDCC99}
algorithms can be seen as a regular collage system.
Furthermore, a collage system is {\em simple}, if it is regular, and
for any variable $X_i = X_\ell X_r$, we have $|X_\ell| = 1$ or $|X_r|
= 1$. 
Output of the LZ78~\cite{LZ78} and LZW~\cite{LZW} algorithms can be seen
as a simple collage system.

To define the derivation tree of a collage system,
we introduce two special symbols $\triangleright$ and $\triangleleft$
that are not in $\Sigma$.
In any sequence over $\Sigma \cup \{\triangleright, \triangleleft\}$,
each symbol $\triangleright$ (resp.\ $\triangleleft$)
``cancels'' the immediately-right (resp.\ -left) symbol in $\Sigma$.
For any assignment $X_i=expr_i$ of a collage system $\mathcal T$,
the \emph{derivation tree} of $X_i$ is a tree with root $v$
labeled $X_i$ such that:
\begin{itemize}
\item $v$ has one subtree consisting of a single node
labeled $a$, if $expr_i = a$ ($a\in\Sigma$).
\item $v$ has two subtrees such that 
the left and the right ones are the derivation trees 
of $X_{\ell}$ and $X_r$, respectively, if $expr_i = X_{\ell} X_r$.
\item $v$ has $p$ subtrees, each of which is the derivation tree
of $X_s$, if $expr_i=(X_s)^p$.
\item $v$ has $(k+1)$ subtrees such that 
the rightmost one is the derivation tree of $X_s$ and 
the others are single-node trees labeled $\triangleright$,
if $expr_i=\prefTrunc{X_s}{k}$.
\item $v$ has $(k+1)$ subtrees such that
the leftmost one is the derivation tree of $X_s$ and
the others are single-node trees labeled $\triangleleft$,
if $expr_i=\sufTrunc{X_s}{k}$.
\end{itemize}
The \emph{derivation tree} of $\mathcal T$ is defined to be
the derivation tree of $X_n$.
Fig.~\ref{fig:DerivationTree} shows the derivation tree
of an example collage system.
We note that
the sequence of leaf-labels of the derivation tree of $\mathcal T$
is a string over $\Sigma\cup\{\triangleright, \triangleleft\}$, and
can be rewritten to $\derive(\mathcal T)$ by
applying the cancellation rules
$\triangleright c \to \varepsilon$ and
$c \triangleleft  \to \varepsilon$
for any character $c\in\Sigma$.
For example, the leaf-label sequence $\mathtt{abcabcabc}\triangleleft\triangleleft \ \mathtt{abc}$
of the derivation tree of Figure~\ref{fig:DerivationTree}
can be rewritten into $\mathtt{cabcaabc}$.

The \emph{size} of a collage system $\mathcal T$ is the number $n$ of
assignments in $\mathcal T$.
Let $\height(X_i)$ represent the height of the derivation tree of $X_i$.
The \emph{height} of a collage system $\mathcal T$, denoted by
$\height(\mathcal T)$, is defined to be $\height(X_n)$.

The \emph{truncated} derivation tree of a collages system $\mathcal{T}$ 
is the tree obtained from the derivation tree of $\mathcal{T}$ as follows: 
(1) a pair of adjacent leaves of form $\triangleright c$ or $c \triangleleft$ is removed ($c \in \Sigma$);
(2) recursively remove internal nodes if they have no children;
(3) repeat until there are no leaves that are labeled with 
$\triangleright$ or $\triangleleft$ in the tree.

We define a function $\derive$ that maps variables $X_i$ to 
strings over $\Sigma$ recursively as follows:
\[
\derive(X_i) = \begin{cases}
  a                         & \text{for $X_i = a$},\\
  \derive(X_{\ell})\derive(X_r)&  \text{for $X_i=X_{\ell}X_r$},\\
  \derive(X_s)^p            &  \text{for $X_i=(X_s)^p$},\\
  \derive(X_s)[k+1:|\derive(X_s)|]   &  \text{for $X_i=\prefTrunc{X_s}{k}$},\\
  \derive(X_s)[1:|\derive(X_s)|-k] &  \text{for $X_i=\sufTrunc{X_s}{k}$}.\\
\end{cases}
\]
A variable $X_i$ is said to \emph{derive} the string $\derive(X_i)$.
Notice that $\derive(X_i)$ is identical to the leaf-label string 
of the subtree of the truncated derivation tree of the collage system
that is rooted at node $X_i$.
A collage system $\mathcal T$ is said to \emph{derive} the string
$T = \derive(X_n)$, i.e., the string derived from the last variable
$X_n$ of $\mathcal T$.
When it is not confusing, we identify a variable $X_i$ with $\derive(X_i)$.
Let $|X_i|=|\derive(X_i)|$ for any variable $X_i$.
$|X_i|$ for all $X_i$ can be computed in a total of $O(n)$ time by a
simple iteration on the variables.
Although $|T|$ can be very large compared to $n$, we shall assume as
in previous work, that the word size is at least $\log |T|$, and hence,
values representing lengths and positions of $T$
in our algorithms can be manipulated in constant time.

For each variable of $X_i$, let 
\begin{itemize}
 \item $\VarOcc(X_i)$ be the number of subtrees rooted at $X_i$ that has exactly $\derive(X_i)$ leaves,
 \item $\trPreVarOcc(X_i)$ be the number of subtrees rooted at $X_i$ such that a non-empty proper prefix of $\derive(X_i)$ is truncated and no non-empty suffix is truncated from its leaf-label string,
 \item $\trSufVarOcc(X_i)$ be the number of subtrees rooted at $X_i$ such that a non-empty proper suffix of $\derive(X_i)$ is truncated and no non-empty prefix is truncated from its leaf-label string,
 \item $\trVarOcc(X_i)$ be the number of subtrees rooted at $X_i$ such that \emph{both} a non-empty proper prefix and a non-empty proper suffix are truncated from its leaf-label string,
\end{itemize}
in the \emph{truncated} derivation tree of a collage system $\mathcal{T}$.
Let $\aVarOcc(X_i)$ denote the number of subtrees rooted at $X_i$ in the 
(non-truncated) derivation tree 
of $\mathcal{T}$.
Let $\dVarOcc(X_i) = \aVarOcc(X_i) - \VarOcc(X_i) - \trPreVarOcc(X_i) - \trSufVarOcc(X_i) - \trVarOcc(X_i)$, i.e., $\dVarOcc(X_i)$ denotes the number of subtrees rooted at $X_i$
in the derivation tree that are completely removed in the truncated derivation tree.
For variable $X_5$ in the running example of Figure~\ref{fig:DerivationTree},
we have $\VarOcc(X_5) = 2$, $\trPreVarOcc(X_5) = 1$, $\trSufVarOcc(X_5) = 1$,
$\aVarOcc(X_5) = 4$, and $\dVarOcc(X_5) = 0$.

For each variable $X_i$ and $1 \leq k < |\derive(X_i)|$,
let $\leaf_i(k)$ denote the leaf of the derivation tree of $X_i$ 
that corresponds to the $k$-th character of $\derive(X_i)$.
In the running example of Figure~\ref{fig:DerivationTree},
$\leaf_9(6)$ is the 6th leaf of the truncated derivation tree 
that corresponds to $\derive(X_9)[6] = \mathtt{a}$.
For string $\derive(X_i)=WYZ$,
the leaves that correspond to $W$ are said to be {\em prefix leaves},
the leaves that correspond to $Y$ are said to be {\em substring leaves},
and the leaves that correspond to $Z$ is said to be {\em suffix leaves}.

\begin{figure}
 \centerline{
  \includegraphics[scale=0.75]{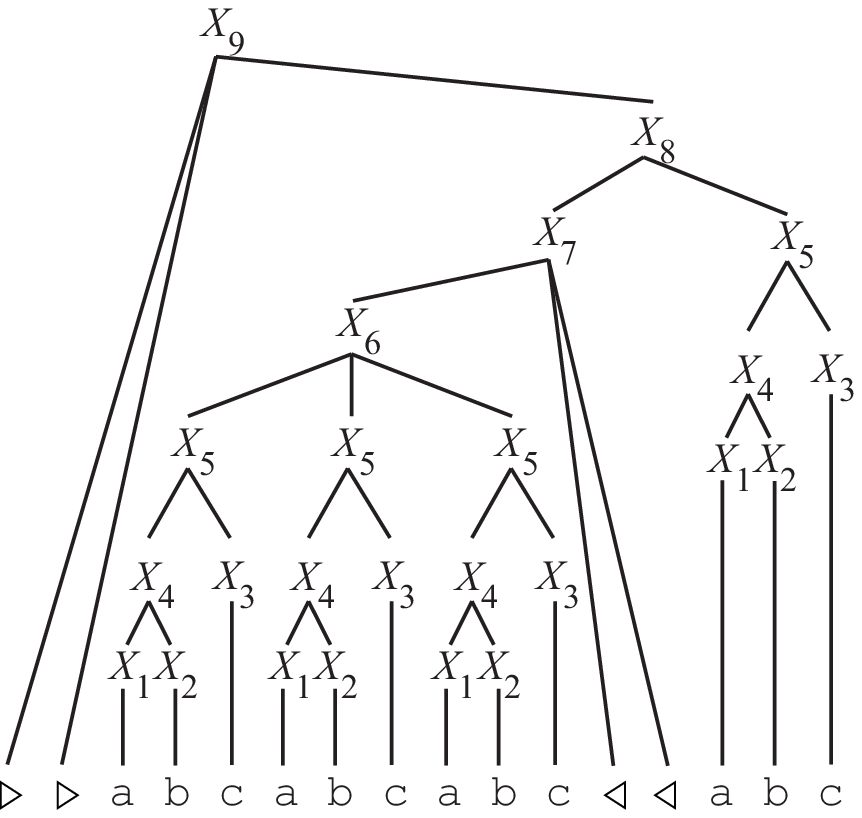}
  \hspace{0.1\textwidth}
  \includegraphics[scale=0.75]{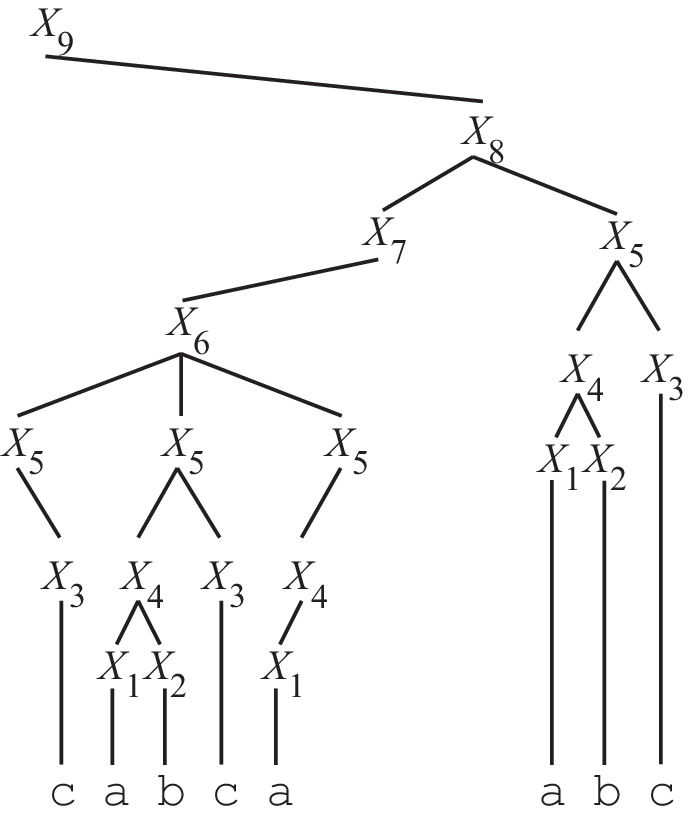}
 }
\caption{
  Derivation tree (left) and truncated derivation tree (right)
  of collage system:
  $\{X_1 = \mathtt{a}, X_2 = \mathtt{b}, X_3 = \mathtt{c}$, 
  $X_4 = X_1 X_2, X_5 = X_4 X_3, X_6 = (X_5)^3, X_7 = \sufTrunc{X_6}{2}, X_8 = X_7 X_5, X_9 = \prefTrunc{X_8}{2}\}$,
  which represents string $\mathtt{cabcaabc}$.
}
\label{fig:DerivationTree}
\end{figure}

\section{Computing $q$-gram Frequencies on Collage Systems}

The main problem we consider in this paper is the following:
\begin{problem}[$q$-gram frequencies on collage systems] \label{prob:q-gram_freq_on_collage}
  Given a collage system ${\mathcal T}$ that describes string $T$,
  compute $|Occ(T, P)|$ for all $q$-grams $P \in \Sigma^q$.
\end{problem}
For regular collage systems (SLPs), a simple and practically efficient
$O(qn)$ time and space algorithm was
recently developed~\cite{goto10:_fast_minin_slp_compr_strin}.
The basic idea is to construct, in $O(qn)$ time, a new string $T'$ of
length $O(qn)$ and an integer array $w$ of the same length so that
$\sum_{j \in Occ(T', P)} w[j] = |\Occ(T,P)|$ for all $P\in\Sigma^q$
where $\Occ(T,P)\neq\emptyset$.

We briefly describe the idea below:
for each $q$-gram occurrence in the text, we identify with it,
the lowest variable in the derivation tree of $\mathcal{T}$, which contains the
$q$-gram occurrence. Thus, we have that each $q$-gram occurrence
corresponds to a unique variable $X_i = X_\ell X_r$ such that
the $q$-gram \emph{crosses} the boundary between $X_\ell$ and $X_r$.
Noticing that all $q$-grams that are identified with $X_i$ are
contained in the string $t_i = \suffix(\derive(X_\ell),q-1) \prefix(\derive(X_r),q-1)$,
consider array $w_i$, with $w_i[j] = \VarOcc(X_i)$ for $1 \leq j \leq |t_i|-q+1$,
and $w_i[j] = 0$ for $|t_i|-q+2 \leq j \leq |t_i|$,
where $\VarOcc(X_i)$ is the number of nodes in the derivation tree
with label  $X_i$\footnote{Note that the derivation tree and the truncated derivation tree of 
any truncation-free collage system are identical.
Hence $\VarOcc(X_i) = \aVarOcc(X_i)$ and $\trPreVarOcc(X_i) = \trSufVarOcc(X_i) = \trSufVarOcc(X_i) = \trVarOcc(X_i) = \dVarOcc(X_i) = 0$ trivially hold.}.
This gives us that
$\sum_{j \in Occ(t_i, P)} w_i[j] = \VarOcc(X_i)\cdot|\Occ(t_i,P)|$
is the total number of occurrences of $q$-gram $P$ in $T$
that are identified with $X_i$, for all $P\in\Sigma^q$.
It remains to sum these values for all $n$ variables, 
that is, $|\Occ(T,P)| = \sum_{i=1}^n \VarOcc(X_i)\cdot|\Occ(t_i,P)|$.
Thus, Problem~\ref{prob:q-gram_freq_on_collage}
reduces to the following problem on
$T'=t_1\cdots t_n$ and  $w = w_1\cdots w_n$:
\begin{problem}[weighted $q$-gram frequencies] \label{prob:weighted-q-gram_frequencies}
  Given a string $T'$, an integer $q$, and integer array $w$ ($|w| = |T'|$),
  compute $\sum_{j \in Occ(T', P)} w[j]$ for all $q$-grams $P \in
  \Sigma^q$ where $\Occ(T',P)\neq\emptyset$.
\end{problem}
(Actually, $t_i$ and $w_i$ for which $|t_i| < q$ can be safely
ignored when constructing $T'$ and $w$.)
Since Problem~\ref{prob:weighted-q-gram_frequencies} is solvable in $O(|T'|)$ time
using standard string indices such as suffix arrays~\cite{manber93:_suffix},
Problem~\ref{prob:q-gram_freq_on_collage} can be solved in $O(qn)$
time and space.

Our algorithm for more general collage systems will follow this
approach of~\cite{goto10:_fast_minin_slp_compr_strin},
but with new challenges lying in the construction of $T'$ and $w$.
First, we show how to adapt the algorithm to cope with repetitions,
and then go on to describe how to further extend the algorithm to cope
with truncations.
 
\subsection{Truncation-Free Collage Systems} \label{sec:truncation_free}

\begin{theorem}
  Problem~\ref{prob:q-gram_freq_on_collage} can be solved in $O(qn)$
  time and space,
  if the collage system $\mathcal{T}$ is truncation-free.
\end{theorem}
\begin{proof}
The strings $\prefix(X_i,d)$ for all variables $X_i$
can be computed in $O(dn)$ time and space using the following dynamic
programming recursion:
Let the array $\prefDP_d[i]$ hold the value of $\prefix(X_i,d)$.
\[
\begin{array}{l}
  
  \prefDP_d[i] 
  = \begin{cases}
    a & \mbox{for } X_i = a\\
    
    \prefDP_d[\ell] & \mbox{for } X_i = X_\ell X_r \mbox{ with } d \leq |X_\ell|, \\

    \prefDP_d[\ell]\cdot
    \prefix(\prefDP_d[r],d-|X_\ell|) & \mbox{for } X_i = X_\ell X_r \mbox{ with } |X_\ell| < d, \\
    (\prefDP_d[s])^y\cdot\prefix(\prefDP_d[s],d-|X_s|y) & \mbox{for } X_i = (X_s)^p, \\
  \end{cases}
\end{array}
\]
where $y = \lfloor d/|X_s| \rfloor$.
(Note that for $X_i=(X_s)^p$, we have $y=0$ and $(\prefDP_d[s])^y = \varepsilon$ when
$|X_s|\geq d$.)
Similarly, the strings $\suffix(X_i, d)$ for all variables $X_i$
can also be computed in $O(dn)$ time and space by dynamic
programming on array $\sufDP_d[i]$.

$\VarOcc(X_i)$ for all $1\leq i\leq n$ can be computed in $O(n)$
time by a simple iteration on the variables,
since $\VarOcc(X_n) = 1$ and for $i < n$,
$\VarOcc(X_i) =
\sum \{ \VarOcc(X_j) \mid X_j = X_\ell X_i \}
+
\sum \{\VarOcc(X_j) \mid X_j = X_iX_r \}
+
\sum \{\VarOcc(X_j) \cdot p \mid X_j = (X_i)^p \}
$

As mentioned previously, we extend the idea
of~\cite{goto10:_fast_minin_slp_compr_strin}
for regular collage systems
so that it handles repetitions.
For each $q$-gram occurrence in the text, 
we identify the lowest variable in the derivation tree of
$\mathcal{T}$, which contains the $q$-gram occurrence. 
For each variable of form $X_i = X_\ell X_r$ with $|X_i| \geq q$,
$t_i$ and $w_i$ are defined as in the case of regular collage systems.
For each variable of form $X_i=(X_s)^{p}$ with $|X_i| \geq q$,
there are two cases:
\begin{enumerate}
\item If $q \leq |X_s|$, then let $t_i = \suffix(\derive(X_s), q-1)\prefix(\derive(X_s), q-1)$.
There exist $p-1$ copies of $t_i$ which cross the boundary of $X_s$'s within $X_i$.
Let $w_i$ be an integer array of length $|t_i|$ such that
$w_i[j] = \VarOcc(X_i)\cdot(p-1)$ for $1 \leq j \leq |t_i|-q+1$,
and $w_i[j] = 0$ for $|t_i|-q+2 \leq j \leq |t_i|$.

\item
  If $|X_s| < q$, then let

  $t_i=\derive(X_s)\prefix(\derive(X_s)^{p-1}, q-1)$,
  which can easily be obtained in $O(q)$ time, given $\prefix(\derive(X_s),q-1)$.
  Let $y = |X_s| - ((q-1) \mod |X_s|)$.
  Then, for $1 \leq j \leq y$, $t_i[j:j+q-1]$ occurs 
  $p - \lceil q / |X_s| \rceil+1$ times in $X_i$,
  and hence we let $w_i[j] = \VarOcc(X_i) \cdot (u - \lceil q / |X_s| \rceil+1)$.
  For $y < j \leq |X_s|$, $t_i[j:j + q-1]$ occurs $p - \lceil q / |X_s| \rceil$ times
  in $X_i$,
  and hence we let $w_i[j] = \VarOcc(X_i) \cdot (p - \lceil q / |X_s| \rceil)$.
  For $|X_s| < j \leq |t_i|$, we let $w_i[j] = 0$.
\end{enumerate}

Now we construct a string $z$ by concatenating each $t_i$ 
with $q \leq |t_i| \leq 2(q-1)$,
and its corresponding weight array $w$ by concatenating 
each $w_i$ with $q \leq |w_i| \leq 2(q-1)$.
Then the problem is reduced to Problem~\ref{prob:weighted-q-gram_frequencies}
on string $z$ and weight array $w$.
The $0$'s inserted at the last parts of each $w_i$ 
avoid to count unwanted $q$-grams generated by the
concatenation of $t_i$ to $z$, which are not substrings of each $t_i$.
Since $|z| = |w| \leq 2(q-1)n$,
the problem can be solved in $O(qn)$ time.
\qed
\end{proof}

Algorithm~\ref{algo:collage_qgram_rep} in appendix 
shows a pseudo-code of our algorithm that solves Problem~\ref{prob:q-gram_freq_on_collage}
for a given truncation-free collage system.

\subsection{General Collage Systems}
We show an $O((q+h)n)$ time and $O(qn)$ space algorithm to solve
Problem~\ref{prob:q-gram_freq_on_collage} for arbitrary general
collage systems, where $h$ is the height of the collage system.

\subsubsection{The $\trPrePath$ and $\trSufPath$ functions}
For variable $X_i=\prefTrunc{X_s}{k}$,
the path from $X_s$ to the leaf $\leaf_i(k+1)$ in the derivation tree of
$X_s$ is called the {\em prefix truncation path} of $X_i$.
For variable $X_i=\prefTrunc{X_s}{k}$,
and $0 \leq x \leq \height(X_i)$,
let $\trPrePath_x(X_i)$ be a function that returns
triple $(X_{u(x)}, \trPre_x, \trSuf_x)$ where
$X_{u(x)}$ is the $x$-th node in the prefix truncation path,
and $X_{u(x)}[\trPre_x+1:|X_{u(x)}|-\trSuf_x]$ corresponds to the prefix of
$X_i[k+1:|X_i|]$ that is derived from this $X_{u(x)}$ in the derivation tree.
Note that the value $|X_{u(x)}|-\trPre_x - \trSuf_x$ is monotonically
non-increasing.

For variable $X_i=\prefTrunc{X_s}{k}$,
we can recursively compute $\trPrePath_x(X_i)$, as follows:
Let $\trPrePath_0(X_i) = (X_i, k,0)$, and for $x \geq 0$ let
\[
\begin{array}{lll}
  \trPrePath_{x+1}(X_i) \\
  =
  \begin{cases}
    (X_\ell, \trPre_{x}, \max(0, \trSuf_{x}-|X_r|)) & \begin{array}{l} \mbox{if } X_{u(x)}=X_\ell X_r  \\
    \quad \mbox{and } 0 \leq \trPre_{x} < |X_\ell|, \end{array} \\
    (X_r, \trPre_{x}-|X_\ell|, \trSuf_{x}) & \begin{array}{l} \mbox{if } X_{u(x)}=X_\ell X_r \\ \quad \mbox{and } |X_\ell| \leq \trPre_{x}, \end{array} \\
    \begin{array}{l} (X_e, \trPre_{x} \!\!\!\!\mod |X_e|, \\
      \quad \max\{0, \lceil \frac{\trPre_{x}}{|X_e|} \rceil \cdot |X_e|+\trSuf_{x}-|X_u|\} \!\!\!\!\mod |X_e|) \end{array}
         & \mbox{if } X_{u(x)}=(X_e)^p, \\
    (X_e, \trPre_{x} + k^\prime, \trSuf_{x}) & \mbox{if } X_{u(x)}=\prefTrunc{X_e}{k^\prime}, \\
    (X_e, \trPre_{x}, \trSuf_{x}+k^\prime) & \mbox{if } X_{u(x)}=\sufTrunc{X_e}{k^\prime}, \\
    \mbox{undefined} & \mbox{if } X_{u(x)} = a, \\
  \end{cases}
\end{array}
\]
where $\trPrePath_{x}(X_i)=(X_{u(x)}, \trPre_{x}, \trSuf_{x})$.
For instance, see Figure~\ref{fig:DerivationTree}.
There, $\trPrePath_x(X_9)$ for $0 \leq x \leq 5$ are 
respectively $(X_8, 2, 0)$, $(X_7, 2, 0)$, $(X_6, 2, 2)$, $(X_5, 2, 0)$, 
and $(X_3, 0, 0)$.

For variable $X_i=\sufTrunc{X_s}{k}$ and its {\em suffix truncation path},
$\trSufPath_x(X_i)$ can be defined and computed analogously.

\subsubsection{Computing length $q-1$ prefixes and suffixes of $\derive(X_i)$}
For all variables $X_i$ and positive integer $d$, 
let the array $\prefDP_d[i]$ (resp. $\sufDP_d[i]$)
hold the value of $\prefix(X_i,d)$ (resp. $\suffix(X_i,d)$).
The strings $\prefix(X_i, d)$ and $\suffix(X_i, d)$
can be computed in a total of $O((d+h)n)$ time and $O(dn)$ space
using a dynamic programming recursion on $\prefDP_d[i]$ and
$\sufDP_d[i]$\footnote{Unlike with truncation-free collage systems,
  $\prefDP_d[i]$ and $\sufDP_d[i]$ are not calculated independently.}.
The cases where $X_i = a$, $X_i = X_\ell X_r$ and $X_i=(X_s)^p$ were mentioned in Section~\ref{sec:truncation_free}.
If $X_i = \sufTrunc{X_s}{k}$,
then $\prefDP_d[i] = \prefix(\prefDP_d[s],|X_i|)$.
Let us now consider 
the case where $X_i = \prefTrunc{X_s}{k}$.
If $|X_i| \leq d$, $\prefDP_d[i]=\suffix(\sufDP_d[s], |X_i|)$.
Otherwise, $|X_i|>d$.
From the monotonicity of $|X_{u(x)}|-\trPre_x-\trSuf_x$,
there exists a unique integer $x$ such that
$X_{u(x)}$, $X_{u(x+1)}$ are descendants of $X_i$ where
$(X_{u(x)}, \trPre_x, \trSuf_x)=\trPrePath_x(X_i)$,
$(X_{u(x+1)},$ $\trPre_{x+1}, \trSuf_{x+1})=\trPrePath_{x+1}(X_i)$,
$|X_{u(x)}|-\trPre_x-\trSuf_x\geq d$,
$|X_{u(x+1)}|-\trPre_{x+1}-\trSuf_{x+1} < d$,
and $X_{u(x)}$ is a concatenation or repetition.
This means that $\prefix(X_i,d)$ crosses the boundary of the children
of $X_{u(x)}$ and can be represented by their suffix and prefix.
Thus, using this $X_{u(x)}$, we have for $X_i = \prefTrunc{X_s}{k}$,

\[ 
   \prefDP_d[i] = \begin{cases}
   \suffix(\sufDP_d[\ell], |X_\ell|-\trPre_x) \cdot \prefix(\prefDP_d[r], d-(|X_\ell|-\trPre_x)) & \mbox{if } X_{u(x)}=X_\ell X_r, \\

   \suffix(\sufDP_d[e], \alpha) \cdot (\prefDP_d[e])^{\beta} 
   \cdot \prefix(\prefDP_d[e], (\trPre_x + d) \!\!\!\mod |X_e|) & \mbox{if } X_{u(x)}=(X_e)^p, \\
  \end{cases}
\]
where $\alpha = |X_e| - (\trPre_x \!\! \mod |X_e|)$ and 
$\beta = \lfloor (d - (|X_e|-(\trPre_x \!\!\mod |X_e|)))/|X_e| \rfloor$.
The corresponding variable $X_u(x)$ can be found in $O(h)$ time.
$\sufDP_d[i]$ can be calculated analogously.
Since $\prefix(X_i, d)$ and $\suffix(X_i,d)$ are strings of length at most $d$,
$\prefix(X_i, d)$ and $\suffix(X_i,d)$ can be computed in a total of
$O((d+h)n)$ time and $O(dn)$ space for all variables $X_i$.

\subsubsection{Computing $\VarOcc(X_i)$}
Here, we describe how the values of 
$\aVarOcc(X_i)$, $\VarOcc(X_i)$, $\trPreVarOcc(X_i)$, $\trSufVarOcc$ $(X_i)$,
$\trVarOcc(X_i)$, and $\dVarOcc(X_i)$ are computed for each variable $X_i$.

\begin{figure}
\centerline{\includegraphics[scale=1.0]{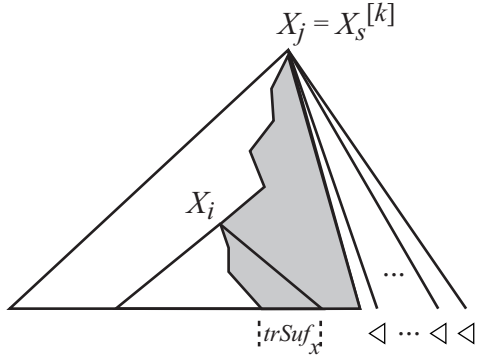}}
\caption{
   The path between the white and gray area is the suffix truncation path for $X_s$.
   $X_i$ lies in this path and 
   the suffix of $X_i$ of length $\trSuf_x > 0$ is truncated in the truncated derivation tree of $X_j$.
}
\label{fig:trSuftr}
\end{figure}

Let $\trSuftr(X_i)$ be the set of pairs $(X_j, d)$ such that
$X_j = \sufTrunc{X_s}{k}$, $X_i = X_{u(x)}$ and $d = \trSuf_x > 0$
for some $x \geq 0$,
where $\trSufPath_x(X_j) = (X_{u(x)}, \trPre_x, \trSuf_x)$.
See also Figure~\ref{fig:trSuftr}.
The suffix truncation path of $X_s$ can contain at most one node 
that is labeled with $X_i$,
and hence there is at most one such value $x$ for each pair of $i$ and $j$.
Also, the first elements of any two pairs in $\trSuftr(X_i)$ are distinct,
and therefore the size of $\trSuftr(X_i)$ does not exceed $n$.

Consider a conceptual $n \times n$ table $D$ such that 
\[
D[i,j] = 
 \begin{cases}
  \trSuf_x & \mbox{if $X_j = \sufTrunc{X_s}{k}$, $X_i = X_{u(x)}$ and $\trSuf_x > 0$
for some $x \geq 0$}, \\
  0 & \mbox{otherwise}.
 \end{cases}
\]
Obviously, the number of non-zero elements in each row $i$ does not exceed $n$.
On the other hand, the number of non-zero elements in each column $j$
does not exceed $\height(X_j)$ (see Figure~\ref{fig:trSuftr}).
Hence the total number of non-zero elements in $D$ does not exceed $nh$,
which means that $\sum_{i=1}^n |\trSuftr(X_i)| \leq nh$.

We can compute $\trSuftr(X_i)$ for all $X_i$ in a total of $O(nh)$ time,
where $h$ is the height of the collage system.
After that, we sort each $\trSuftr(X_i)$ in increasing order of 
the second value of the pairs in $\trSuftr(X_i)$.
The total time cost to sort $\trSuftr(X_i)$ for all $X_i$ is 
\[O(\sum_{i=1}^n|\trSuftr(X_i)|\log|\trSuftr(X_i)|) = O(nh \log n).\]
The $l$-th element of $\trSuftr(X_i)$ is denoted by 
$\trSuftr(X_i)[l]$ for $1 \leq l \leq |\trSuftr(X_i)|$.

$\trPretr(X_i)$ can be defined and computed analogously.

\begin{lemma} \label{lem:truncation_vOcc}
Let $\mathcal{T} = \{X_i = expr_i\}_{i=1}^n$ be a general collage system.
Assume that, for all variables $X_i = \prefTrunc{X_s}{k}$ and 
$X_{i^\prime} = \sufTrunc{X_{s^\prime}}{k^\prime}$,
$\trSuftr(X_i)$ and $\trPretr(X_{i^\prime})$ are already computed with their elements sorted.
Then, we can compute 
$\VarOcc(X_i)$,
$\trPreVarOcc(X_i)$,
$\trSufVarOcc(X_i)$,
 $\trVarOcc(X_i)$,
$\dVarOcc(X_i)$,
and $\aVarOcc(X_i)$
for all variables $X_i$ in a total of $O(nh)$ time,
where $h$ is the height of $\mathcal{T}$.
\end{lemma}

\begin{proof}
Clearly $\VarOcc(X_n) = \aVarOcc(X_n) = 1$ and 
$\trPreVarOcc(X_n) =$ $\trSufVarOcc(X_n)=$ $\trVarOcc(X_n) =$
$\dVarOcc(X_n)$ $=$ $0$.

Suppose that, for $i \leq n$,
we have already computed $\VarOcc(X_{i^\prime})$, 
$\aVarOcc(X_{i^\prime})$,$\trPreVarOcc(X_{i^\prime})$,
$\trSufVarOcc(X_{i^\prime})$, $\trVarOcc(X_{i^\prime})$,
$\dVarOcc(X_{i^\prime})$, $\trPretr(X_{i^\prime})$, and $\trSuftr(X_{i^\prime})$ 
for all $i \leq {i^\prime} \leq n$.
We propagate some those values to the descendants of $X_i$ as follows:

If $X_i = X_\ell X_r$, then 
there are also $\aVarOcc(X_i)$ occurrences of $X_\ell$ in the derivation tree.
Thus we increase $\aVarOcc(X_\ell)$ by $\aVarOcc(X_i)$.
There are also $\dVarOcc(X_i)$
occurrences of $X_\ell$ that are completely truncated in the truncated derivation tree.
Thus we increase $\dVarOcc(X_\ell)$ by $\dVarOcc(X_i)$.
$\aVarOcc(X_r)$ and $\dVarOcc(X_r)$ are computed similarly.
This takes a total of $O(n)$ time for all $X_i = X_\ell X_r$.

If $X_i = (X_s)^p$, then there are $p \cdot \aVarOcc(X_i)$ occurrences of 
$X_e$ in the derivation tree,
and there are $p \cdot \dVarOcc(X_i)$ occurrences of $X_e$ that are completely truncated
in the truncated derivation tree.
Thus we increase 
$\aVarOcc(X_s)$ and $\dVarOcc(X_s)$ by $p \cdot \aVarOcc(X_i)$ and
$p \cdot \dVarOcc(X_i)$, 
respectively.
This takes a total of $O(n)$ time for all $X_i = (X_s)^p$.

If $X_i = \prefTrunc{X_s}{k}$, then we increase
$\aVarOcc(X_s)$ and $\dVarOcc(X_s)$ by $\aVarOcc(X_i)$ and $\dVarOcc(X_i)$,
respectively.
For $x \geq 0$, 
let $\trPrePath_x(X_i) = (X_{u(x)}, \trPre_x, \trSuf_x)$.
Consider the path $X_{u(0)} = X_s$, $X_{u(1)}$, \ldots, $X_{u(v)}$,
where $v$ is the largest integer satisfying $\trPre_x > 0$.
By the definition of $\trPrePath$, we know that $\trPre_x > 0$
for any $0 \leq x \leq v$.
Since $\trPre_x + \trSuf_x < |\derive(X_{u(x)})|$,
we do not increase the value of $\dVarOcc(X_{u(x)})$ at this time.
We increase $\trPreVarOcc(X_{u(x)})$ if $\trSuf_x = 0$,
and $\trVarOcc(X_{u(x)})$ if $\trSuf_x > 0$,
by $\VarOcc(X_i)$, respectively.
Now we consider the nodes that lie on the left of the path.
If $X_{u(x)}$ is of form $X_{u(x)} = X_\ell X_r$ and $\trPre_x \geq |X_\ell|$,
then $X_\ell$ is completely truncated 
in the truncated derivation tree.
Hence we increase $\dVarOcc(X_\ell)$ by $\VarOcc(X_i)+\trSufVarOcc(X_i)$.
If $X_{u(x)}$ is of form $X_{u(x)} = (X_e)^p$,
then the first $\lfloor \trPre_x / |X_e| \rfloor$
repetitions of $X_e$ are completely truncated,
and hence we increase $\dVarOcc(X_\ell)$ by $\lfloor \trPre_x / |X_e| \rfloor
\cdot (\VarOcc(X_i) + \trSufVarOcc(X_i))$.

Further care is taken for the occurrences of $X_i$
whose non-empty suffix is truncated due to its ancestor corresponding to $\trSuftr(X_i)$,
as follows:
For each $1 \leq l \leq |\trSuftr(X_i)|$,
let $(X_{j(l)}, d(l)) = \trSuftr(X_i)[l]$, where $X_{j(l)} = \sufTrunc{X_p}{g}$.
By definition, on the suffix truncation path of $X_p$
there exists a subtree rooted at $X_i$ whose suffix of length $d(l)$ is truncated.
A key observation is that the nodes, which lie on the prefix truncation path 
of $X_i$ but do \emph{not} lie on the suffix truncation path of $X_{j(l)}$,
have $\VarOcc(X_{j(l)})$ occurrences in the truncated derivation tree of $X_{j(l)}$.
Let $W_l$ be the subset of these nodes
which consists of the nodes whose non-empty prefix is truncated.
For each variable $Y \in W_l$, 
either $\trPreVarOcc(Y)$ or $\trVarOcc(Y)$ has to be 
increased by $\VarOcc(X_{j(l)})$ accordingly.
For each fixed $X_i$, we have to do this for all the ancestors of $X_i$ 
corresponding to $\trSuftr(X_i)$.
If this is done separately for each ancestor, 
it takes a total of $O(n^2 h)$ time for all $i$.
We can however speed up this by processing elements of $\trSuftr(X_i)$ 
in increasing order of $d(l)$:
For each $1 \leq l \leq |\trSuftr(X_i)|$,
we propagate $\sum_{m=l}^{|\trSuftr(X_i)|} \VarOcc(X_{j(m)})$ to the nodes in $W_l - W_{l+1}$
(see also Figure~\ref{fig:trPretrSufvOcc}),
where we let $W_{|\trSuftr(X_i)|+1} = \emptyset$ for simplicity.
For each fixed $X_i$, this can take $O(n)$ time.
However, the overall time complexity is $O(nh)$ for all $X_i$,
since $\sum_{i=1}^{n}|\trSuftr(X_i)| = O(nh)$ as stated previously.
For the nodes that lie on the left of the prefix truncation path of $X_i$,
we increase their $\dVarOcc$ value by $\sum_{m = 1}^{|\trSuftr(X_i)|} \VarOcc(X_{j(m)})$.
This can also be done in $O(nh)$ time.

\begin{figure}[t]
  \centerline{
    \includegraphics[scale=1.0]{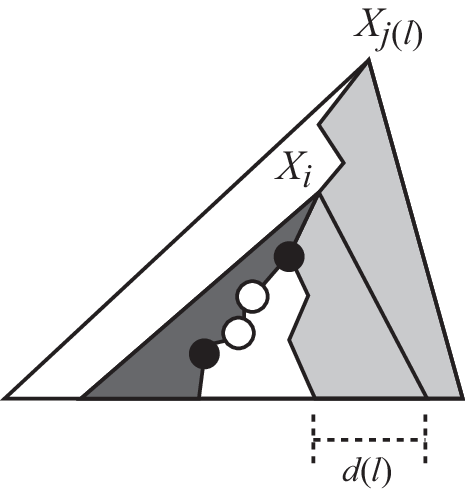}
    \hfil
    \includegraphics[scale=1.0]{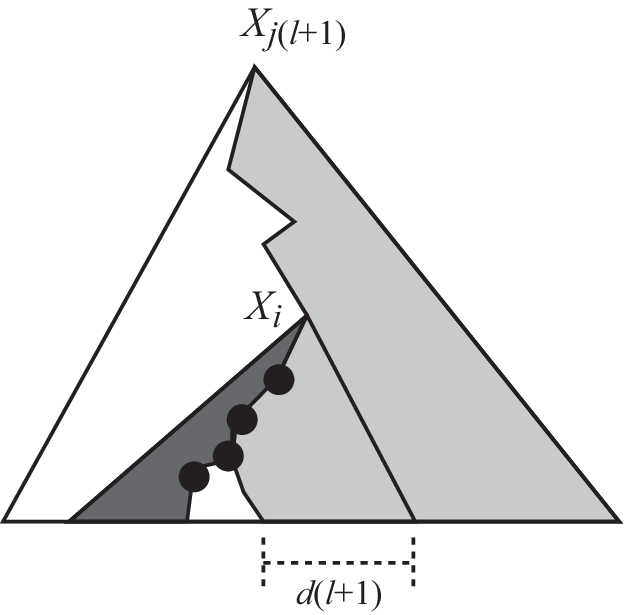}
  }
  \caption{
   The circles represent nodes (i.e. variables) that lie on the 
   prefix truncation path of $X_i = \prefTrunc{X_s}{k}$.
   The white circles in the left diagram represent nodes in $W_l - W_{l+1}$.
   For each $Y \in W_l - W_{l+1}$,
   we increase either $\trPreVarOcc(Y)$ or $\trVarOcc(Y)$
   by $\sum_{m = l}^{|\trSuftr(X_i)|} \VarOcc(X_{j(m)})$,
   depending on if a non-empty suffix of $Y$ is truncated or not.
  }
  \label{fig:trPretrSufvOcc}
\end{figure}

If $X_i = \sufTrunc{X_s}{k}$, then the values are propagated similarly 
in case of $X_i = \prefTrunc{X_s}{k}$, in a total of $O(nh)$ time.
\qed
\end{proof}

Algorithm~\ref{algo:vOcc_on_collage} in appendix shows 
a pseudo-code of our algorithm
to compute $\VarOcc(X_i)$,\\
$\trPreVarOcc(X_i)$,$\trSufVarOcc(X_i)$,
$\trVarOcc(X_i)$, $\dVarOcc(X_i)$, and $\aVarOcc(X_i)$.

\subsubsection{Construction of weight array}
As with truncation free collage systems,
we again consider reducing Problem~\ref{prob:q-gram_freq_on_collage}
to Problem~\ref{prob:weighted-q-gram_frequencies} of computing
weighted $q$-gram frequencies on a single uncompressed string.
For each $q$-gram occurrence in the text, we again identify
the lowest variable in the truncated derivation tree of $\mathcal{T}$,
which contains the $q$-gram occurrence.
Observe that, 
in this strategy no $q$-grams will be identified with a truncation variable $X$,
as there always exists a non-truncation descendant of $X$
with which the corresponding $q$-grams are identified.
Thus we construct string $t_i$ for variable $X_i=X_\ell X_r$ and $X_i=(X_s)^p$,
as in Section~\ref{sec:truncation_free}, and
it remains to set the value of $w_i[j]$ so that it represents
the total number of occurrences of the $q$-gram in the text,
corresponding to $t_i[j:j+q-1]$ derived by $X_i$.

\begin{figure}
\centerline{\includegraphics[width=0.4\textwidth]{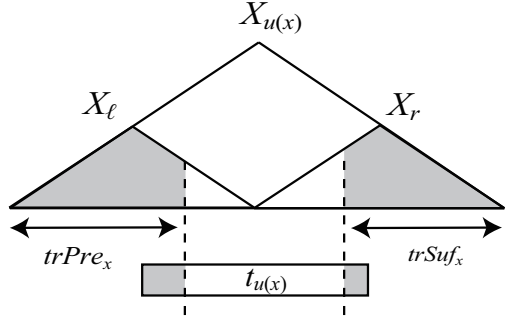}}
\caption{
    A non-empty truncated prefix and a possibly non-empty truncated suffix of 
    $X_{u(x)} = X_\ell X_r$
    are shown in gray.
    The weights for $w_{u(x)}$ are set accordingly
    for the white range of $t_{u(x)}$. 
  }
  \label{fig:weight_concatenation}
\end{figure}

Firstly, we consider complete (i.e. non-truncated) occurrences of variable $X_i$
in the truncated derivation tree of the collage system.
By definition, there are $\VarOcc(X_i)$ such occurrences, 
and hence we set the weights for $w_i$ in a similar way to Section~\ref{sec:truncation_free}.

Secondly, we consider the occurrences of $X_i$ 
where a non-empty prefix and/or non-empty suffix of the leaf-label string 
of the subtree rooted at $X_i$ is truncated in the truncated derivation tree 
of the collage system.
Consider a variable $X_y = \prefTrunc{X_s}{k}$ with $y > i$
and let $v \geq 0$ be the largest integer satisfying $\trPre_v > 0$,
where $\trPrePath_v(X_j) = (X_{u(v)}, \trPre_v, \trSuf_v)$.
Assume 
that there exists an integer $0 \leq x \leq v$ such that 
$u(x) = i$, where $\trPrePath_x(X_y)=(X_{u(x)}, \trPre_x, \trSuf_x)$.
This implies that $X_i$ lies on the prefix truncation path of $X_y$
and a non-empty prefix of $X_{i}$ is truncated 
in the truncated derivation tree of $X_y$.
We have the following cases depending on the type of $X_{u(x)}$ (recall $u(x) = i$):

If $X_{u(x)}=X_\ell X_r$, there are two sub-cases:
(1) If $\trPre_x \geq |X_\ell|$ or $\trSuf_x \geq |X_r|$,
then no $q$-grams are identified with this occurrence of $X_{u(x)} (= X_i)$.
(2) If $\trPre_x < |X_\ell|$ and $\trSuf_x < |X_r|$,
then $X_{u(x)} (= X_i)$ derives a string $\derive(X_\ell)[\trPre_x+1:|X_\ell|]\cdot \derive(X_r)[1:|X_r|-\trSuf_x])$.
Then string $t_{u(x)}[\max(1, \trPre_x-\max(0, |X_\ell|-q+1)+1):\min(q-1, |X_\ell|)-\max(0, \trSuf_x + q-1-|X_r|)+q-1]$ crosses the boundary of $X_\ell$ and $X_r$,
so we increase the weight of $w_{u(x)}[j]$ by $\VarOcc(X_i)$ for
each $j$, where
$\max(1, \trPre_x-\max(0, |X_\ell|-q+1)+1) \leq j \leq \min(q-1, |X_\ell|)-\max(0, \trSuf_x + q-1-|X_r|)$.
See also Figure~\ref{fig:weight_concatenation}.

If $X_{u(x)}=(X_e)^p$,
let $r=p-\lfloor \trPre_x / |X_s| \rfloor-\lfloor \trSuf_x / |X_s| \rfloor-2$,
this occurrence of $X_{u(x)}$ derives string
$\derive(X_e)[(\trPre_x \!\! \mod |X_e|)+1:|X_e|] \cdot 
\derive(X_e)^{\max\{0,r\}} \cdot 
\derive(X_e)[1:(\trSuf_x \!\!\mod |X_e|)-1]$.
In what follows we consider the case where $r>0$ and $|X_e|<q\leq|X_e|^r$.
Let $g=|X_e|-((q-1) \mod |X_e|)$.
There are four types of occurrences of $q$-gram $t_u[j:j+q-1]$:
$t_u[j:j+q-1]$ occurs $(r- \lceil q/|X_e| \rceil + 1)$ times for $1 \leq j \leq g$,
$t_u[j:j+q-1]$ occurs $(r- \lceil q/|X_e| \rceil)$ times for $g < j < q$,
within the $(X_e)^r$ term.
$t_u[j:j+q-1]$ occurs crossing the boundary of 
$\prefTrunc{X_e}{\trPre_x \mod |X_e|}$ and $(X_e)^r$ for $(\trPre_x \mod |X_e|) < j \leq |X_e|$.
$t_u[j:j+q-1]$ occurs crossing the boundary of $(X_e)^r$ and
$\sufTrunc{X_e}{\trSuf_x \mod |X_e|}$ for $1 \leq j \leq |X_e|-((\trSuf_x+q-1) \mod |X_e|)$.
We can set the weights of $w_{u(x)}$ for each of the 4 above ranges of $j$, accordingly.
For example,
if $X_{u(x)}=(X_e)^9$, $\trPre_x=4$, $\trSuf_x=5$,
$\derive(X_e)=\mathtt{aba}$ and $q=5$, then we have $t_{u(x)}=\mathtt{abaabaa}$ 
and $w_{u(x)}=[4,5,5,0,0,0,0]$. (See also Figure~\ref{fig:weight_rep}.)
For the other cases, we can compute the weights similarly.
Note that there are $O(h)$ variables in the prefix truncation path of 
$X_y = \prefTrunc{X_s}{k}$.
This may lead to $O(qnh)$ time complexity,
as the total length of the $w$ array is $O(qn)$.
We can however reduce the time cost to $O((q+h)n)$ using 
a differential representation $\witv$ of $w$ 
such that $w[j] = \sum_{l = 1}^{j}\witv[l]$ for every $1 \leq j \leq |w|$.
Given positive integers $b,e$ such that $1 \leq b \leq e \leq |w|$,
increasing the value of $w[j]$ for all $b \leq j \leq e$ by $d$ reduces to 
increasing the value of $\witv[b]$ by $d$ and decreasing the value of $\witv[e+1]$ by $b$,
which can be done in $O(1)$ time.

\begin{figure}[t]
  \centerline{\includegraphics[width=0.7\textwidth]{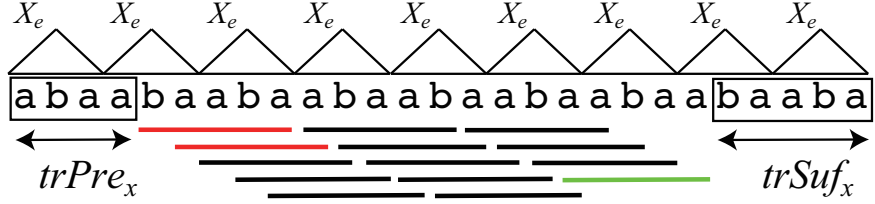}}
  \caption{
    Illustration for $\prefTrunc{X_e}{\trPre_x}(X_e)^5\sufTrunc{X_e}{\trSuf_x}$,
    $\trPre_x=4$, $\trSuf_x=5$.
    Variable $X_e$ derives the string $\mathtt{aba}$,
    and the number of $5$-grams starting inside $\prefTrunc{X_e}{\trPre_x}$ is 2,
    and the number of $5$-grams completely contained within $(X_e)^5$ is 11,
    and the number of $5$-grams ending inside $\sufTrunc{X_e}{\trSuf_x}$ is 1.
  }
  \label{fig:weight_rep}
\end{figure}

For all variables $X_i=X_\ell X_r$ and $X_i=(X_s)^p$,
we can compute weight array $\witv_i$ in $O(n)$.
For all variables $X_i=\prefTrunc{X_s}{k}$ and $X_i=\sufTrunc{X_s}{k}$,
we can compute weight arrays $\witv_{u(x)}$ for all variables $X_{u(x)}$ 
in the prefix or suffix truncation path of $X_i$, in $O(hn)$ time.
Then $w$ can be obtained by a simple scan of $\witv$ in $O(qn)$ time.

Now, we construct a string $z$ by concatenating each $t_i$ with
$q \leq |t_i| \leq 2(q-1)$, and its corresponding weight array $w$
by concatenating each $w_i$ with $q \leq |w| \leq 2(q-1)$.
Then, Problem~\ref{prob:q-gram_freq_on_collage} for a general
collage system reduces to Problem~\ref{prob:weighted-q-gram_frequencies} of
weighted $q$-gram frequencies on a single uncompressed string,
and hence we obtain:
\begin{theorem}
  Problem~\ref{prob:q-gram_freq_on_collage} can be solved in $O((q+h\log n)n)$
  time and $O(qn)$ space,
  for general collage systems.
\end{theorem}

\newcommand{\weightRep}{weightRep}
\newcommand{\weightAlgo}{weightAlgo}
\newcommand{\weightNaive}{weightNaive}
\newcommand{\weight}{\weightNaive}

\bibliographystyle{splncs03}
\bibliography{ref}

\clearpage
\appendix
\section{Appendix}

\begin{algorithm2e}[h]
  \caption{Calculating $q$-gram frequencies of a truncation-free
    collage system for $q\geq 2$}
  \label{algo:collage_qgram_rep}
  \SetKw{KwAnd}{and}
  \SetKwInput{KwOut}{Report}
  \KwIn{SLP ${\mathcal T} = \{X_i\}_{i=1}^n$ representing string $T$, integer $q\geq 2$.}
  \KwOut{all $q$-grams and their frequencies which occur in $T$.}
  \SetKw{KwReport}{Report}
  
  Calculate $\VarOcc(X_i)$ for all $1\leq i\leq n$\; \label{algo:slpmain:varocc}  
  Calculate $\prefix(\derive(X_i),q-1)$ and $\suffix(\derive(X_i),q-1)$ for all $1\leq
  i\leq n-1$ \; \label{algo:slpmain:prefsuf}  
  $z \leftarrow \varepsilon$; $w \leftarrow []$\;
  \For{$i \leftarrow 1$ \KwTo $n$}{\label{algo:slpmain:mainloop}
    \If{$|X_i| \geq q$}{
      \If{$X_i = X_{\ell}X_r$ \KwAnd $|X_i| \geq q$}{
        $t_i = \suffix(\derive(X_\ell),q-1)\prefix(\derive(X_r),q-1)$ \;
        $w_i \leftarrow $ create integer array of length $|t_i|$, each element set to 0 \;
        \lFor{$j \leftarrow 1$ \KwTo $|t_i|-q+1$}{
          $w_i[j] \leftarrow \VarOcc(X_i)$ \;
        }
      }
      \ElseIf{$X_i = (X_s)^p$ \KwAnd $|X_s| \geq q$}{
        $t_i = \suffix(\derive(X_s),q-1)\prefix(\derive(X_s),q-1)$ \;
        $w_i \leftarrow $ create integer array of length $|t_i|$, each element set to 0 \;
        \lFor{$j \leftarrow 1$ \KwTo $|t_i|-q+1$}{
          $w_i[j] \leftarrow \VarOcc(X_i) \cdot (p-1)$\;
        }
      }
      \ElseIf{$X_i = (X_s)^p$ \KwAnd $|X_s| < q$}{
        $t_i = \prefix(\derive(X_s)^{\min\{p,\lceil (|X_s|+q-1)/|X_s|\rceil\}}, |X_s|+q-1)$ \;
        $w_i \leftarrow $ create integer array of length $|t_i|$, each element set to 0 \;
        $y = |X_s| - ((q-1) \mod |X_s|)$ \;
        \lFor{$j \leftarrow 1$ \KwTo $y$}{
          $w_i[j] \leftarrow \VarOcc(X_i) \cdot (p - \lceil q / |X_s| \rceil+1)$\;
        }
        \lFor{$j \leftarrow y+1$ \KwTo $|X_j|$}{
          $w_i[j] \leftarrow \VarOcc(X_i) \cdot (p - \lceil q / |X_s| \rceil)$\;
        }
      }
      $z$.append($t_i$)\; \label{algo:slpmain:zappend}
      $w$.append($w_i$)\; \label{algo:slpmain:append0}
    }
  }
  \KwReport  
  $q$-gram frequencies in $z$, where each $q$-gram $z[i:i+q-1]$
  is {\em weighted} by $w[i]$.\label{algo:slpmainweightedfreqs}
\end{algorithm2e}

\begin{algorithm2e}[t]
  \caption{Calculate $\VarOcc(X_i)$ for all variables of general collage system}
  \label{algo:vOcc_on_collage}
  \SetKw{KwBreak}{break}
  \SetKw{KwTrue}{true}
  \SetKw{KwFalse}{false}
  \SetKw{KwAnd}{and}
  \SetKw{KwOr}{or}
  \KwIn{A general collage system $\mathcal{T} = \{X_i\}_{i=1}^n$}
  \KwOut{$\VarOcc(X_i)$ for all $1 \leq i \leq n$}

  compute $\trPretr(X_i)$, $\trSuftr(X_i)$ for all variable $X_i$ \;

    Initialize the values of $\aVarOcc(X_i)$, $\VarOcc(X_i)$,
    $\trPreVarOcc(X_i)$, $\trSufVarOcc(X_i)$, $\trVarOcc(X_i)$,
    $\dVarOcc(X_i)$ to 0 for all $X_i$ \;
  $\aVarOcc(X_n) \leftarrow 1$ \;
  \For{$i \leftarrow n$ \KwTo $1$}{
    
    $\VarOcc(X_i) \leftarrow \aVarOcc(X_i)-\dVarOcc(X_i)-\trVarOcc(X_i)-\trPreVarOcc(X_i)-\trSufVarOcc(X_i)$ \;
    
    \If{$X_i = X_\ell X_r$}{
      $\aVarOcc(X_\ell) \leftarrow \aVarOcc(X_\ell) + \aVarOcc(X_i)$;
      $\aVarOcc(X_r) \leftarrow \aVarOcc(X_r) + \aVarOcc(X_i)$ \;
      $\dVarOcc(X_\ell) \leftarrow \dVarOcc(X_\ell) + \dVarOcc(X_i)$;
      $\dVarOcc(X_r) \leftarrow \dVarOcc(X_r) + \dVarOcc(X_i)$ \;
    }\ElseIf{$X_i=(X_s)^p$}{
      $\aVarOcc(X_s) \leftarrow \aVarOcc(X_s) + p*\aVarOcc(X_i)$ \;
      $\dVarOcc(X_s) \leftarrow \dVarOcc(X_s) + p*\dVarOcc(X_i)$ \;
    }\ElseIf{$X_i=\prefTrunc{X_s}{k}$}{
      
      $\aVarOcc(X_s) \leftarrow \aVarOcc(X_s) + \aVarOcc(X_i)$ ;
      $\dVarOcc(X_s) \leftarrow \dVarOcc(X_s) + \dVarOcc(X_i)$ \;
      $x \leftarrow 0$ ;
      $l \leftarrow 1$ ;
      $\trR \leftarrow 0$ ;
      $\trSsum \leftarrow 0$ \; 
      $(X_{u(x)}, \trPre_x, \trSuf_x) \leftarrow \trPrePath_x(X_s, k)$ \;
      \While{$\trPre_x > 0$}{

        $(X_{j(l)}, d(l)) \leftarrow \trSuftr(X_i)[l]$ \;
        \While{$\trR \leq d(l)$}{
          $\trSsum \leftarrow \trSsum+\VarOcc(X_{j(l)})$ \;
          $l \leftarrow l+1$ ;
          $(X_{j(l)}, d(l)) \leftarrow \trSuftr(X_i)[l]$ \;
        }
        \tcp{propagate $\VarOcc(X_i)+\sum_{m=1}^{l-1} \VarOcc(X_{j(m)})$ to nodes in $W_l-W_{l+1}$.}

        \lIf{$\trSuf_x>0$}{
          $\trVarOcc(X_{u(x)}) \leftarrow \trVarOcc(X_{u(x)}) + \VarOcc(X_i)+\trSsum$ \;
          
        }\lElse{
          $\trPreVarOcc(X_{u(x)}) \leftarrow \trPreVarOcc(X_{u(x)}) + \VarOcc(X_i)+\trSsum$ \;
          
        }

        \If{$X_{u(x)}=X_\ell X_r$}{
          \If{$|X_\ell| \leq \trPre_x$}{
            $\dVarOcc(X_\ell) \leftarrow \dVarOcc(X_\ell) + (\VarOcc(X_i)+\trSufVarOcc(X_i))$ \;
          }\lElse{
            $\trR \leftarrow \trR+|X_r|$ \;
          }
        }\ElseIf{$X_{u(x)}=(X_e)^p$}{
          $\dVarOcc(X_e) \! \leftarrow \! \dVarOcc(X_e) +
          \lfloor \trPre_x / |X_e| \rfloor \! * \! (\VarOcc(X_i)+\trSufVarOcc(X_i))$ \;
          
          $\trR \leftarrow \trR+p-\lceil |\trPre_x|/|X_e|\rceil$ \;
        }
	$x \leftarrow x + 1$ ;
        $(X_{u(x)}, \trPre_x, \trSuf_x) \leftarrow \trPre_xPath_x(X_s, k)$ \;
      }

    }\ElseIf{$X_i=\sufTrunc{X_s}{k}$}{
      \tcp{omitted: analogous to prefix truncation}

    }

  }
\end{algorithm2e}

\end{document}